\documentclass[submission,copyright,creativecommons]{eptcs}

\usepackage{mdframed}
\usepackage{graphicx,bussproofs}
\usepackage{amsmath,amssymb}
\usepackage{latexsym}
\usepackage{hyperref}
\usepackage{stmaryrd}
\usepackage{fancyvrb}
\newcommand{\sbr}[1]{\llbracket #1 \rrbracket}

\newcommand{\centerpic}[2]{
        \begin{center}
                \includegraphics[scale={#1}]{{#2}}
        \end{center}
        }

\usepackage{color}

\usepackage{amsthm}
	
\newtheorem{theorem}{Theorem}
\newtheorem{proposition}[theorem]{Proposition}
\newtheorem{lemma}[theorem]{Lemma}

\newtheorem{example}{Example}

\allowdisplaybreaks

\begin{document}

\title{A Structural and Nominal Syntax for Diagrams}
\author{
	{Dan R. Ghica}
	\institute{University of Birmingham}
	\and
	Aliaume Lopez
	\institute{ENS Cachan, Universit\'e Paris-Saclay}
}
\def\titlerunning{A Structural and Nominal Syntax for Diagrams}
\def\authorrunning{D. R. Ghica and A. Lopez}
\maketitle
%12pp

\begin{abstract}
The correspondence between monoidal categories and graphical languages of diagrams has been studied extensively, leading to applications in quantum computing and communication, systems theory, circuit design and more. From the categorical perspective, diagrams can be specified using (name-free) combinators which enjoy elegant equational properties. However, conventional notations for diagrammatic structures, such as hardware description languages (VHDL, \textsc{Verilog}) or graph languages (\textsc{Dot}), use a  different style, which is flat, relational, and reliant on extensive use of names (labels). Such languages are not known to enjoy nice syntactic equational properties. However, since they make it relatively easy to specify (and modify) arbitrary diagrammatic structures they are more popular than the combinator style. In this paper we show how the two approaches to diagram syntax can be reconciled and unified in a way that does not change the semantics and the existing equational theory. Additionally, we give sound and complete equational theories for the combined syntax. 
\end{abstract}

\section{Specifying graphs}

Graphs and their visual representations (diagrams) are an appealing way of describing many kinds of systems, in particular circuits. Work originated in the study of quantum computation~\cite{DBLP:conf/lics/AbramskyC04} has exploited the connection between various classes of graphs and monoidal categories, going back to the seminal work of Joyal and Street~\cite{joyal1991geometry}, to add a layer of structure which makes reasoning about and with diagrams not just intuitive but also mathematically rigorous. Subsequently this connection was extended in many surprising and interesting directions, from computational linguistics~\cite{DBLP:journals/logcom/SadrzadehCC13}, to modelling signal flow~\cite{DBLP:conf/popl/BonchiSZ15} and  synchronous~\cite{GhicaJung16} or asynchronous~\cite{ghica2013diagrammatic} circuits. New automated reasoning tools based on diagrams, rather than the usual linear algebraic syntax, are a particularly exciting development (see {\url{http://globular.science/}}).  This convenient and elegant interplay between the dual categorical and diagrammatic methods are by now a mature and rich area of research~\cite{selinger2010survey}.

Although these developments convincingly establish the usefulness of categorical diagrammatics in {reasoning} about many kinds of systems, we note that little has been suggested in terms of a workable syntax which is conceptually compatible with it, but also with conventional notations, which are unstructured and relational. Examples of the latter are hardware description languages such as \textsc{Verilog} and VHDL, or graph languages such as \textsc{Dot} (see {\url{http://www.graphviz.org/}}). Indeed, in the literature the categorical combinators are usually taken as an implicit diagram syntax. Whereas such a syntax is expressive enough to describe the desired classes of graphs, it is often not as convenient as the alternatives. Although categorical combinators can be elegant and succinct in certain situation, \textit{point-free} languages of combinators generally hold little appeal as concrete syntax (e.g. APL).

Conventional diagram syntax lacks structure. It is a ``flat'' relational description of the graph. Whereas the categorical combinators are name-free, conventional syntax uses an abundance of names, for nodes and sometimes even for edges. Although these languages are widely used, they are broadly considered both inelegant and unwieldy, and purely structural alternatives have been proposed~\cite{DBLP:conf/icfp/BjesseCSS98}. 

Before diving into technicalities let us consider a simple motivating example. Consider the diagram below, consisting of two components $f$ and $g$, the first with one input and three outputs and second with three inputs and one output, connected in the obvious way:

\centerpic{1}{fg}

In the categorical, combinator-style, specification such a diagram would be succinctly written as the \textit{composition} $f;g: 1\rightarrow 1$, of $f:1\rightarrow 3$ and $g:3\rightarrow 1$. A VHDL-style description would look more verbose:
\begin{Verbatim}
module fg(input u; output v)
begin
  component f(input x1; output y1, y2, y3);
  component g(input z1, z2, z3; output t1);
  wire y1, z1; wire y2, z2; wire y3, z3; wire u,  x1; wire t1, v;
end 
\end{Verbatim}
This is the diagram annotated with all the names used above:
\centerpic{1}{fgnom}

In highly structured diagrams the categorical notation is concise and elegant. However, realistic circuits may also use arbitrary connections. Consider a variation of the circuit above:
\centerpic{1}{fg2}
The flat description can be adjusted to cope with arbitrary connector reassignments:
\begin{Verbatim}
module fg(input u1, u2; output v1, v2)
begin
  component f(input x1; output y1, y2, y3);
  component g(input z1, z2, z3; output t1);
  wire y1, z1; wire y3, z3; wire y2, v2; wire u2, z2; wire u1, x1; wire t1, v1;
end 
\end{Verbatim}

The categorical style description is now more intricate, requiring a mix of sequential ($;$) and parallel (\textit{tensor}, $\otimes$) composition along with the use of \textit{structural} combinators such as the single connector (\textit{identity}, $i$) and crossing connector (\textit{symmetry}, $\gamma$): 
$
(f\otimes i);(i\otimes \gamma \otimes i); (i\otimes i\otimes \gamma);(i\otimes \gamma \otimes i);(g\otimes i)
.$

The combinator-style variable-free syntax is appealing for highly structured diagrams and awkward for arbitrary ones, whereas the flat and unstructured syntax seems useful in the case of unstructured and unnecessarily verbose for structured graphs. 
Without advocating one style or the other, in this paper we  present a syntax which shows that the two  are compatible. Our contributions are therefore as follows:
\begin{enumerate}
	\item Preserving and extending the equational theory.
	\item Possibly eliminating the need for  structural combinators.
	\item Preserving and conserving the expressiveness of the underlying diagrams.
	\item Reasoning equationally in the new syntax.
\end{enumerate}

\section{Uniflow diagrams}

A PROP (abbreviation of \textit{products and permutations}) is a strict symmetric monoidal category where every object is a natural number~\cite{Hackney2015}. We give a graph semantics of PROPs based on Kissinger's \textit{framed point graphs}~\cite{Kissinger}. 
Let a labelled directed acyclic graph (DAG) be a DAG $(V, E)$ equipped with a partial injection $f : V \rightharpoonup L$ and a relation $E \subseteq V^2$ such that the transitive and reflexive closure of $E$ is a partial order on $V$. Let a labelled interfaced DAG (LIDAG) be a labelled DAG with two distinguished lists of unlabelled nodes representing the ``input'' and ``output'' ports. Unlabelled nodes are called \textit{wire nodes}, and edges connecting them are called \textit{wires}. 
A wire homeomorphism~\cite[Sec.~5.2.1]{Kissinger} is any insertion or removal of wire nodes along wires which does not otherwise change the shape of the graph. 
%The rewrite rules are:
%\centerpic{1}{homeo}
Two LIDAGs are considered to be equivalent if they are graph isomorphic up to renaming vertices and wire homeomorphisms. The quotienting of LIDAGs by this equivalence gives us \textit{framed point DAGs}~\cite[Def. 5.3.1]{Kissinger}, which we will call \textit{uniflow diagrams}.
We give a syntax for uniflow diagram as follows:
\[
M::= k \mid i \mid \gamma \mid M;M \mid M\otimes M \mid x \mid \overline{xy}.M,
\]
where $k\in K$ are the constants and $x$ variables. 
The language is essentially that of symmetrical monoidal categorical combinators (identity, symmetry, composition, tensor) over a signature, extended with variables and a binding construct we read as \textit{link $x$ and $y$ in $M$}. 
We equip this language with a type system, where judgements are of the form 
$\Gamma \mid \Delta \vdash M : m \rightarrow n\mid{R}. $
$\Gamma$ is a set of \textit{input variables}, $\Delta$ of \textit{output variables} and $R$ a partial order on their (disjoint) union, called the \textit{anchor} of the diagram (Fig.~\ref{fig:udtr}). The anchor relation is a syntactic discipline used to prevent the inadvertent introduction of cycles. Note that $i:1\rightarrow 1$ and $\gamma:2\rightarrow 2$ have the same typing rules as the constants.  The type $m\rightarrow n$ represents a uniflow diagram with $m$ (unlabelled) inputs and $n$ (unlabelled) outputs. Given a relation $R$ we denote by $[R]$  its transitive and reflexive closure. For any set $S$, we define 
$R\setminus S=\{(x,y)\in R\mid x,y\not\in S\}.$
\begin{figure*}
	\small
	\centering
	\AxiomC{}
	\UnaryInfC{$x \mid - \vdash x : 0\rightarrow 1\mid x\leq x $} 
	\DisplayProof\quad
	\AxiomC{}
	\UnaryInfC{$-\mid x  \vdash x : 1\rightarrow 0\mid x\leq x$} 
	\DisplayProof
	\quad
	\AxiomC{}
	\UnaryInfC{$- \mid - \vdash k : m\rightarrow n\mid\emptyset$} 
	\DisplayProof
	\\[1.5ex]
	\AxiomC{$\Gamma \mid \Delta \vdash M: m_1\rightarrow m_2\mid {\leq}$}	
	\AxiomC{$\Gamma' \mid \Delta' \vdash N: n_1\rightarrow n_2\mid {\leq'}$}
	\RightLabel{$m_2=n_1$}
	\BinaryInfC{$\Gamma\uplus\Gamma' \mid \Delta\uplus\Delta' \vdash M;N : m_1\rightarrow n_2\mid [{\leq}\cup{\leq'}\cup((\Gamma\uplus\Delta)\times(\Gamma'\uplus\Delta')]$}
	\DisplayProof
	\\[1.5ex]	
	\AxiomC{$\Gamma \mid \Delta \vdash M: m_1\rightarrow m_2\mid {\leq}$}	
	\AxiomC{$\Gamma' \mid \Delta' \vdash N: n_1\rightarrow n_2\mid {\leq'}$}
	\BinaryInfC{$\Gamma\uplus\Gamma' \mid \Delta\uplus\Delta' \vdash M\otimes N : m_1+n_1\rightarrow m_2+n_2\mid [{\leq}\cup{\leq'}]$}
	\DisplayProof
	\\[1.5ex]	
	\AxiomC{$\Gamma\mid\Delta\vdash M:m_1\rightarrow m_2\mid {\leq}$}
	\AxiomC{$[{\leq}\cup\{(x,y)\}]$ is a partial order}
	\RightLabel{$x,y\in\Gamma\uplus\Delta$}
	\BinaryInfC{$\Gamma\mid\Delta\vdash M:m_1\rightarrow m_2\mid {[{\leq}\cup\{(x,y)\}]} $}
	\DisplayProof
	\\[1.5ex]
	\AxiomC{$\Gamma,x \mid\Delta, y\vdash M: m_1\rightarrow m_2\mid {\leq}$}	
	\AxiomC{$x\leq y$}
	\BinaryInfC{$\Gamma\mid\Delta  \vdash \overline{xy}.M:m_1\rightarrow m_2\mid {{\leq}\setminus\{x,y\}}$}
	\DisplayProof
	\caption{Uniflow diagrams typing rules}
	\label{fig:udtr}
\end{figure*}

We give a {concrete graph-theoretic semantics} of the uniflow diagram. 
We may write $\Gamma\mid\Delta\vdash M$ if $\Gamma\mid\Delta\vdash M:m\rightarrow n\mid R$ for some $m,n\in\mathbb{N}$ and some partial order $R$.
The semantics is defined by induction on the typing derivation. The meaning of a judgement $\Gamma\mid\Delta\vdash M:m\rightarrow n\mid {R}$ is a uniflow diagram $\bigl(V,I, O, E, f:V\rightharpoonup K\uplus\mathbb{A}\bigr)$ such that $\mathbb A\supseteq\Gamma,\Delta$ is a set of names and the relation $R$ is a partial order on the vertices labelled by variables $\Gamma\uplus\Delta$ which is compatible with $E$, i.e. if $(x, y)\in R$ then there exist unique vertices $u,v$ such that $f(u)=x, f(v)=y$ and \textit{there is no path} in $E$ from $v$ back to~$u$. We write this as
$
\sbr{\Gamma\mid\Delta\vdash M:m\rightarrow n\mid R}
=(V,I, O, E, f)\asymp R.
$ 
The structural combinators are interpreted as in Fig.~\ref{fig:sceq}.

\begin{figure*} 
	\small
\begin{align*}
\sbr{x\mid-\vdash x:0\rightarrow 1\mid x\leq x} 
= (\{a,b\},nil, b::nil, \{(a,b)\}, \{(a,x)\})
\asymp \{(a,a)\}
=\raisebox{-2ex}{\includegraphics[]{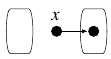}}\\
\sbr{-\mid x\vdash x:1\rightarrow 0\mid x\leq x} 
= (\{a,b\},a::nil, nil, \{(a,b)\}, \{(b,x)\})
\asymp\{(b,b)\}
=\raisebox{-2ex}{\includegraphics[]{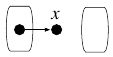}}\\
\sbr{-\mid -\vdash i:1\rightarrow 1\mid \emptyset} = (\{a,b\},a::nil, b::nil, \{(a,b)\}, \emptyset)
\asymp\emptyset
=\raisebox{-2ex}{\includegraphics[]{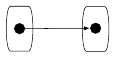}}\\
\sbr{-\mid -\vdash \gamma:2\rightarrow 2\mid \emptyset} = (\{a,a',b,b'\},a::a'::nil, b::b'::nil, \{(a,b'),(a',b)\}, \emptyset)\asymp \emptyset
=\raisebox{-3ex}{\includegraphics[]{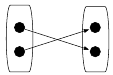}}
\end{align*}
\caption{Interpreting structural combinators}
\label{fig:sceq}
\end{figure*}
Let $::$ stand for list concatenation or cons-ing. And let the ``zipping'' of two lists be defined in the usual way (${zip}\;nil\;nil=nil$, and ${zip}\,(x::xs)\,(y::ys)=(x,y)::({zip}\;xs\;ys)$). If unambiguous we may use a list to also mean the set of elements of that list. 

Assuming that  
$\sbr{\Gamma_i \mid \Delta_i \vdash M_i: m_i\rightarrow n_i\mid {R_i}}= (V_i,I_i,O_i, E_i, f_i)\asymp R_i, \text{ for } (i=1,2),
$
The interpretation of the two forms of composition is: 
\begin{align*}
&\sbr{\Gamma_1\uplus\Gamma_2 \mid \Delta_1\uplus\Delta_2 \vdash M;N : m_1\rightarrow n_2 \mid {[{R_1}\cup{R_2}\cup(\Gamma_1\uplus\Delta_1)\times(\Gamma_2\uplus\Delta_2)]}}\\
&\qquad\qquad =(V_1\uplus V_2, I_1, O_2, E_1\cup E_2\cup ({zip}\;O_1\; I_2),f_1\cup f_2)
\asymp [{R_1}\cup{R_2}\cup(\Gamma_1\uplus\Delta_1)\times(\Gamma_2\uplus\Delta_2)]\\
&\sbr{\Gamma_1\uplus\Gamma_2 \mid \Delta_1\uplus\Delta_2 \vdash M\otimes N : m_1+m_2\rightarrow n_1+n_2 \mid {[{R_1}\cup{R_2}]}}\\
&\qquad\qquad=(V_1\uplus V_2, I_1:: I_2, O_1:: O_2, E_1\cup E_2, f_1\cup f_2)
\asymp{[{R_1}\cup{R_2}]}.
\end{align*}
%Sequential composition and tensor can be visualised as in Fig.~\ref{fig:comp}. 
Note that since diagrams are defined up to relabelling nodes, these definitions are always good. In case of name clashes nodes can be given fresh names.
%\begin{figure*}
%	\centerpic{1}{figcomp}
%	\caption{Sequential and tensor composition}\label{fig:comp}
%\end{figure*}

Constants are interpreted as below, with $p_i, q_j$ labels for the input and output ports of the constant.
\begin{equation*}
\sbr{-\mid-\vdash k:m\rightarrow n\mid \emptyset} =\raisebox{-5ex}{\includegraphics[]{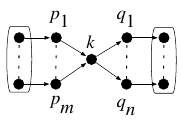}}
\end{equation*}
 %This is an equivalent alternative to ``box-and-wire'' languages of diagrams used for monoidal categories. 

The rule for strengthening the anchor is interpreted as
\begin{equation*}
\sbr{\Gamma\mid\Delta\vdash M:m\rightarrow n\mid {[{R}\cup\{(x,y)\}]}} =
(V,I,O, E, f)
\asymp[{R}\cup(f^{-1}(x),f^{-1}{(y)})].
\end{equation*}
The restriction on $[{R}\cup(f^{-1}(x),f^{-1}{(y)})]$ being a proper partial order ensures that the interpretation of the rule is well defined. 

Let $(f\setminus B')(x)=f(x)$ if $f(x)\not\in B'$ and undefined otherwise.
The new construct is the \textit{link}:
\begin{equation}
\sbr{\Gamma\mid\Delta  \vdash\overline{xy}.M:m\rightarrow n\mid {{R}\setminus\{x,y\}} } =
(V,I,O, E\cup(f^{-1}(x),f^{-1}{(y)}), f\setminus\{x,y\})
\asymp{R}\setminus\{(x,y)\}.
\label{eqn:link}
\end{equation}
The link can be visualised as below, including the wire homeomorphism:
\centerpic{1}{link}
%Up to wire homeomorphism, the second diagram is equivalent to
%\centerpic{1}{linkdm2}
\begin{lemma}[Soundness]
	Any well typed term denotes a valid uniflow diagram. 
\end{lemma}
%\begin{proof}
%	Immediate, by induction on derivation. 
%\end{proof}

The useful connection between diagrams and monoidal categories is given by this proposition:
\begin{proposition}\label{prop:diagcat}
	Given a term $\Gamma\mid\Delta \vdash M:m\rightarrow n$, the diagram $\sbr{\Gamma\mid\Delta \vdash M:m\rightarrow n}$ is a morphism in the strict free symmetric monoidal category (SMC) over signature $K\cup \Gamma\cup\Delta $ of constants and (uninterpreted) variables. 
\end{proposition}
By ``uninterpreted variables'' we mean simply that the variables from $\Gamma,\Delta$ are morphism variables in the categorical signature with types $0\rightarrow 1$ and $1\rightarrow 0$ respectively. 
This is just Thm.~3.12 from~\cite{selinger2010survey}, which is restating the classic result from~\cite{joyal1991geometry}. For frame-point graphs this is Thm.~5.5.10 in~\cite{Kissinger}.

Note that this does not make the semantics categorical, since the \textit{link} construction is not defined categorically but only concretely. However, any term (with our without link) corresponds to a diagram which can be described categorically.
%\textbf{An important note}: we cannot deem our semantics to be categorical, even though each diagram is a morphism in this SMC, with sequential composition as morphism composition and tensorial composition as monoidal tensor. The reason is that  \textit{link} is not interpreted by a categorical construct in the SMC. A proper categorical semantics will be discussed in the concluding section but left as further work. Although a categorical semantics would be certainly desirable, the concrete diagrammatic intepretation is enough for soundness and completeness. And even in the absence of a categorical semantics, the connection between diagrams and monoidal categories will prove useful in enabling the proof of definability!
An immediate consequence of Prop.~\ref{prop:diagcat} is that all derivations of a well typed term produce the same uniflow diagram, possibly with a different anchoring relation. 
\begin{lemma}
The type system of uniflow diagrams is coherent up to the anchoring relation. 
\end{lemma}

Note  that different derivations may require different anchoring relations, even though the underlying diagrams are equal as FPGs. 
\begin{lemma}[Definability]\label{lem:defdag}
	Any uniflow diagram is definable just in terms of composition, tensor, link.
\end{lemma}
\begin{proof}
	We introduce exactly two link nodes on each connector (a wire homeomorphism) and we label them with fresh variable names. Suppose there are $j$ wires and $k$ boxes, $m$ inputs and $n$ outputs. Let us denote by $u_i$ where $i=1,n$ the labels associated with input connectors and by $v_i$ where $i=1,m$ the labels associated with the output connectors. 
	Let $y_{i,j}$ (respectively $z_{i,j}$) be the labels of the $j$th input (respectively output) for each constant occurrence $k_i$. Let 
	$
	\vec u=\bigotimes_{i=1,n}u_i,
	\vec v=\bigotimes_{i=1,m}v_i,\ 
	\vec{y_i}=\bigotimes_{j=1,\text{dom}(k_i)} y_{i,j},
	\vec{z_i}=\bigotimes_{j=1,\text{codom}(k_i)} z_{i,j}. $
	Let $M_0=\vec u; \left(\bigotimes_{i=1,k} \vec y_i; k_i; \vec z_i\right); \vec v$. 
	The term is 
$	M=\overline{x_1x_1'}\cdots\overline{x_nx'_n}.M$. A link $\overline {x_ix'_i}$ is created if there is an edge between the nodes labelled by $x_i$ and $x'_i$ in $M_0$, with $x_i,x'_i$ chosen from the variables introduced above. $M$ should be a closed form.  
	It is straightforward to show that this is indeed denotes the desired diagram. 
	We only need to prove that this is well typed. The subterm $\bigotimes_{i=1,k} \vec y_i; f_i;\vec z_i$ is clearly well typed. Moreover, we can always extend the order with any $x_i\leq x_i'$ before linking without creating cycles because the starting graph is a DAG. 
\end{proof}
The construction in the proof is just the flat nominal notation used by conventional graph languages. 
\begin{theorem}
	For any term $M$ there exist terms $\tilde M$ and $ \hat M$ such that $\tilde M$ is link-free and $\hat M$ is free of structural combinators (identity, symmetry) and 
	$
	\sbr M=\sbr{\tilde M}=\sbr{\hat M}.
	$
\end{theorem}
%\begin{proof}
%	Given a term $\Gamma\mid\Delta\vdash M:m\rightarrow n$ we construct the uniflow diagram $\sbr{M}$. The uniflow diagrams is a morphism in the free strict symmetric monoidal category over the signature $K$ extended with variables $\Gamma,\Delta$ (Prop.~\ref{prop:diagcat}) so the existence of $\tilde M$ is immediate. The existence of $\hat M$ is given by Lem.~\ref{lem:defdag}. 
%\end{proof}

\begin{example}\label{ex:gam}
	Symmetry at any type, e.g.~$\gamma_{2,2}:4\rightarrow 4$
%	\centerpic{1}{c22dm}
	can be defined in  combinatorial or  nominal style:
	\begin{align*}
	\gamma_{2,2}&=\overline{aa'}.\overline{bb'}.\overline{cc'}.\overline{dd'}.a\otimes b\otimes c\otimes d\otimes c'\otimes d'\otimes a'\otimes b'
	=(i\otimes\gamma\otimes i);(\gamma\otimes \gamma);(i\otimes\gamma\otimes i).
	\end{align*}
\end{example}
Any $\gamma_{m,n}$ for $m,n\neq 0$ can be defined from $\gamma$, and identity $id_m$ at any non-zero type  can be defined from~$i$. We sometimes write the identity $id_m$ as just $m$. 
\begin{example}
	We can now revisit our introductory example. The most succinct description mixes variables and structural connectors: 
$	\overline{x'x}.\overline{yy'}.(f\otimes y);(i\otimes x'\otimes y'\otimes i);(g\otimes x).$
%	visualised as
%	\centerpic{1}{fg2link}
\end{example}

%\subsection{Equational theory for links}
It can be easily checked that all equational properties of the underlying PROP are carried over to the syntax. We interpret $\Gamma\mid\Delta \vdash M\equiv N$ as the equality of the uniflow diagrams represented by $\sbr{\Gamma\mid\Delta\vdash M}$ and $\sbr{\Gamma\mid\Delta\vdash N}$ (ignoring the anchoring order). The equations are now parametrised by the set of free variables, as in Fig.~\ref{fig:moncat}, assuming integers $m,n,$ etc. are chosen so that the terms are well-typed.
\begin{figure*}[h]
\begin{align}
	\Gamma\mid\Delta \vdash (M_1;M_2);M_3 & \equiv M_1;(M_2;M_3) & \text{associativity of composition}\\
	\Gamma\mid\Delta \vdash id_m;M & \equiv M;id_n\equiv M &\text{identity}\\
	\Gamma\mid\Delta \vdash (M;M')\otimes (N;N') & \equiv (M\otimes N);(M'\otimes N') &\text{tensor functoriality}\\
	\Gamma\mid\Delta \vdash (M_1\otimes M_2)\otimes M_3 & \equiv M_1\otimes (M_2\otimes M_3) &\text{strictness}\\
	\Gamma\mid\Delta \vdash M\otimes 0 & \equiv 0\otimes M \equiv M &\text{strictness}\\
	\Gamma\mid\Delta \vdash (M_1\otimes M_2);\gamma_{n_1,n_2} &\equiv \gamma_{m_2,m_1};(M_2\otimes M_1) &\text{symmetry}\\
	\vdash \gamma_{0,1}&\equiv\gamma_{1,0}\equiv 1 &\text{strict symmetry}\\
	\vdash \gamma_{m,n+p}&\equiv(\gamma_{m,n}\otimes p);(n\otimes\gamma_{m,p}) &\text{strict symmetry}\\[1.5ex] 
	\Gamma\mid\Delta \vdash \overline{xy}.M&\equiv \overline{x'y'}.M\{x'y'/xy\}\quad x',y'\text{ fresh}\label{eq:alpha} &\text{alpha-equivalence}\\
	\Gamma\mid\Delta \vdash \overline{xy}.(M;N) &\equiv (\overline{xy}.M) ; N  \quad x,y\not\in \mathit{fv}(N) &\text{scope extrusion}\label{eq:binder1} \\
	\Gamma\mid\Delta \vdash \overline{xy}.(M\otimes N) &\equiv (\overline{xy}.M) \otimes N  \quad x,y\not\in \mathit{fv}(N) &\text{scope extrusion}\label{eq:binder12} \\
	\Gamma\mid\Delta \vdash \overline{xy}.(M;N) &\equiv M ; \overline{xy}.N,\quad x,y\not\in \mathit{fv}(M)&\text{scope extrusion} \label{eq:binder2}\\
	\Gamma\mid\Delta \vdash \overline{xy}.(M\otimes N) &\equiv M \otimes \overline{xy}.N,\quad x,y\not\in \mathit{fv}(M)&\text{scope extrusion} \label{eq:binder22}\\
	\vdash \overline{xy}.x; y &\equiv 1 &\text{link identity}\label{eq:gamma1}
\end{align}
\caption{Equational theory of uniflow diagrams}
\label{fig:moncat}
\end{figure*}

We can now turn our attention to the equational theory of uniflow diagrams~(Fig.~\ref{fig:moncat}). Since $\overline{xy}.M$ is a binding construct, we expect alpha-equivalence to hold, which is reflected in a new axiom, Eqn.~\ref{eq:alpha}. 
Variable renaming $M\{x/y\}$ and the set of free variables $\mathit{fv}(M)$ are defined in the obvious way. For name management we also have new equations  for scope extrusion (Eqn.~\ref{eq:binder1}--\ref{eq:binder22}). Finally, the connection between variables and the structural connectors, namely identity, is given in Eq.~\ref{eq:gamma1}.
 
\begin{theorem}\label{thm:unisc}
	The equational theory of uniflow diagrams is sound and complete for the given semantics. 
\end{theorem}
\begin{proof}
	The soundness is immediate. For completeness, consider two terms such that $\sbr {\Gamma\mid\Delta \vdash M}=\sbr {\Gamma\mid\Delta \vdash N}$. We first rename all bound variables using alpha-equivalence, then, using Eqns.~\ref{eq:binder1} and~\ref{eq:binder2} we move all binders into global scope. In case the resulting terms are link-free the equation $\Gamma\mid\Delta\vdash M\equiv N$ is provable because the terms denote morphism in the free strict symmetric monoidal category over signature $K$ extended with the uninterpreted variables $\Gamma,\Delta$ (Prop.~\ref{prop:diagcat}). 	
	Suppose that one of the terms has an outermost link, so it has form $\Gamma\mid\Delta\vdash\overline{xy}.M$. We reason diagrammatically about term $\Gamma,x\mid\Delta,y\vdash M$, which is a morphism in the free SMC over $K$ extended with 	$\Gamma\uplus\Delta\uplus \{x,y\}$. We can always redraw the diagram to an isomorphic diagram in 
	which the nodes labelled by $x$ and $y$ are adjacent: 
	
	\centerpic{1}{linkadj}
	
	This diagrammatic equivalence can be proved in the equational theory of the uniflow diagrams, noting that $\gamma_{m,n}$ and the identities at $m$ can be defined in terms of $\gamma$ and $i$. Moreover, the diagram on the right will correspond to a term which has subterm $x;y$, with variables $x,y$ not occurring anywhere else. So we can apply Eqns.~\ref{eq:binder1} and~\ref{eq:binder2} repeatedly until the link binder only binds this subterm, $\overline{xy}.x;y$. Then we use Eqn.~\ref{eq:gamma1} to replace this subterm with the identity. We repeat the process until all binders are eliminated. 
\end{proof}
%
%COHERENCE??

\section{Biflow diagrams}

Let us consider now strict symmetric \textit{traced} monoidal categories and their associated diagrammatic language, \textit{biflow diagrams}~\cite{Cazanescu:1990:TNA:97367.97373}. Biflow diagrams are labelled directed graphs with vertices $V$ and edges $E$, equipped with a partial injection $f:V \rightharpoonup L$ from vertices to labels. Biflow diagrams are labelled directed graphs with interfaces $I,O$, lists of unlabelled ports, equivalent up to vertex renaming and wire homeomorphism. In other words, biflow diagrams are uniflow diagrams minus the anchoring relation, similar to Kissinger's \textit{FPGs}. 
The increased expressiveness of the diagrammatic language is reflected into a relaxation of the type system. The type judgements are $\Gamma\mid\Delta \vdash M:m_1\rightarrow m_2$, similar to the uniflow case but without an anchoring relation. The rules are also similar, with the addition of new rules for \textit{trace} (Fig.~\ref{fig:bftr}).
\begin{figure*}
	\small
	\centering
	\AxiomC{}
	\UnaryInfC{$x \mid - \vdash x : 0\rightarrow 1$} 
	\DisplayProof\quad
	\AxiomC{}
	\UnaryInfC{$-\mid x,  \vdash x : 1\rightarrow 0$} 
	\DisplayProof
	\quad
	\AxiomC{}
	\UnaryInfC{$- \mid - \vdash i : 1\rightarrow 1$} 
	\DisplayProof\quad
	\AxiomC{}
	\UnaryInfC{$- \mid - \vdash \gamma : 2\rightarrow 2$} 
	\DisplayProof
	\\[1.5ex]
	\AxiomC{$\Gamma \mid \Delta \vdash M: m_1\rightarrow m_2$}	
	\AxiomC{$\Gamma' \mid \Delta' \vdash N: n_1\rightarrow n_2$}
	\RightLabel{$m_2=n_1$}
	\BinaryInfC{$\Gamma\uplus\Gamma' \mid \Delta\uplus\Delta' \vdash M;N : m_1\rightarrow n_2$}
	\DisplayProof
	\quad
	\AxiomC{$\Gamma \mid \Delta \vdash M: m_1\rightarrow m_2$}	
	\AxiomC{$\Gamma' \mid \Delta' \vdash N: n_1\rightarrow n_2$}
	\BinaryInfC{$\Gamma\uplus\Gamma' \mid \Delta\uplus\Delta' \vdash M\otimes N : m_1\rightarrow n_2$}
	\DisplayProof
	\\[1.5ex]
	\AxiomC{$\Gamma\mid\Delta \vdash M: m_1\rightarrow m_2$}
	\UnaryInfC{$\Gamma\mid\Delta \vdash \mathrm{Tr}^i (M):m_1-i\rightarrow m_2-i$}
	\DisplayProof
	\quad
	\AxiomC{$\Gamma,x \mid\Delta, y\vdash M: m_1\rightarrow m_2$}	
	\UnaryInfC{$\Gamma\mid\Delta  \vdash \overline{xy}.M:m_1\rightarrow m_2$}
	\DisplayProof
	\caption{Type system for biflow diagrams}
	\label{fig:bftr}
\end{figure*}
If the trace is over the unit we omit the superscript. Traces for $i>1$ need not be included as primitive, as they can be constructed out of unit traces. 

The diagrammatic interpretation is the same as in the previous section, omitting the anchoring order on vertices. The interpretation of the unit trace operator is the standard one, a \textit{feedback} edge between the top input and output ports. Assuming that  
$\sbr{\Gamma \mid \Delta \vdash M: m+1\rightarrow n+1}= (V,i::I,o::O, E, f)$,
$
	\sbr{\Gamma \mid \Delta \vdash \mathrm{Tr}(M) : m\rightarrow n}=
	(V,I,O,E\cup\{(o,i)\}).
$
For equational completeness we also include the (trivial) trace over 0, $\mathrm{Tr}^0(M)=M$. 
Soundness holds, immediately. 
\begin{lemma}[Biflow definability]\label{lem:bifdef}
	Any biflow diagram is definable in terms of composition, tensor, link.
\end{lemma}
%\begin{proof} 
%The proof is the same construction as in Lem.~\ref*{lem:defdag}, leading to the new term: 
%\[
%\overline{x_1x_1'}\cdots\overline{x_nx'_n}.\bigotimes_{i=1,k} \vec y_i f_i\vec z_i.
%\]
%In this case the derivation of the term is obviously possible, since the rules are strictly more relaxed than in the case of DAGs. 
%\end{proof}

\begin{theorem}
	For any term $ M$ in the language of biflow diagrams, there are forms $ \tilde M$ and $\hat M$ such that $\tilde M$ is link-free and $\hat M$ is free of structural combinators (identity, symmetry, trace) and 
	$
	\sbr M=\sbr{\tilde M}=\sbr{\hat M}.
	$
\end{theorem}
%\begin{proof}
%	The proof is immediate from the definability of biflow diagrams (Lem.~\ref{lem:bifdef}). 
%\end{proof}

%\subsection{Equational theory of biflow diagrams}

In the equational theory we inherit all the equations from the previous section, plus the trace equations~\cite{DBLP:conf/tlca/Hasegawa97}. All the additional equations are given in Fig.~\ref{fig:treqns}. Surprisingly, we do not need new equations for link, except for scope extrusion in the presence of trace (Eqn.~\ref{eq:trextr}).

Other relations between trace and link can be derived in the existing equational theory. The proposition below is clearly sound in the model, corresponding to the two alternative descriptions of a back-link, using trace or \textit{link}.
%\centerpic{1}{linktr}
%The proposition can also be derived equationally.
\begin{proposition}\label{prop:linktr}
$	\Gamma\mid\Delta \vdash \mathrm{Tr} (M)
	\equiv\overline{yx}.(x\otimes m);M;(y\otimes n)
$
\end{proposition}
\begin{proof}
	\begin{align*}
	\overline{yx}.(x\otimes m);M;(y\otimes n) 
	&= \overline{yx}.\mathrm{Tr}^0\bigl((x\otimes m);M;(y\otimes n)\bigr) &\text{(vanishing)} \\
	&=\overline{yx}.\mathrm{Tr}\bigl(m;M;(y;x\otimes n)\bigr) &\text{(sliding)} \\
	&=\overline{yx}.\mathrm{Tr}\bigl(M;(y;x\otimes n)\bigr) &\text{(identity)} \\
	&=\mathrm{Tr}\bigl(M;(\overline{yx}.(y;x)\otimes n)\bigr) &\text{(link scope)} \\
	&=\mathrm{Tr}\bigl(M;(1\otimes n)\bigr) &\text{(link)} \\
	&=\mathrm{Tr}\bigl(M\bigr) &\text{(identity)}
	\end{align*}
\end{proof}
\begin{theorem}
	The equational theory of biflow diagram is sound and complete for the given semantics.
\end{theorem}
\noindent
\textbf{Note}: Diagrams corresponding to compact-closed categories can be described analogously. 
\begin{figure*}
\begin{align}
\Gamma\mid\Delta\vdash\mathrm{Tr}^m((g\otimes n);f)&\equiv g;\mathrm{Tr}^m(f) &\text{left-tightening}\\
\Gamma\mid\Delta\vdash\mathrm{Tr}^m(f;(g\otimes n))&\equiv \mathrm{Tr}^m(f);g &\text{right-tightening}\\
\Gamma\mid\Delta\vdash\mathrm{Tr}^m(f;(n\otimes g))&\equiv \mathrm{Tr}^{p}((q\otimes g);f) &\text{sliding}\\
\Gamma\mid\Delta\vdash\mathrm{Tr}^0(f)&\equiv f &\text{vanishing}\\
\Gamma\mid\Delta\vdash\mathrm{Tr}^{m+n}(f)&\equiv\mathrm{Tr}^{m}(\mathrm{Tr}^n(f)) &\text{vanishing}\\
\Gamma\mid\Delta\vdash\mathrm{Tr}^m(n\otimes s)&\equiv n\otimes\mathrm{Tr}^m(f) &\text{superposing}\\
\vdash \mathrm{Tr}^1(\gamma)&\equiv 1 &\text{yanking}\\[1.5ex]
\Gamma\mid\Delta\vdash \overline{xy}.\mathrm{Tr}^i(M)&\equiv\mathrm{Tr}^i(\overline{xy}.M), \quad i\in\{0,1\}. &\text{scope extrusion}\label{eq:trextr}
\end{align}
\caption{Additional equations for biflow diagrams}
\label{fig:treqns}
\end{figure*}
%\begin{proof}
%	The proof is the same as for Thm.~\ref{thm:unisc}, except that the ambient diagrammatic theory is now different. Diagram manipulation required to create subterms $x;y$ si now allowed to involve the use of trace, so we need not assume an ordering on the variables. The additional scope extrusion rule allows all the binders to be moved equationally to outermost scope, so that all terms have shape $\Gamma\mid\Delta\vdash \overline{x_1y_1}\ldots\overline{x_ny_n}.M'$ with $M'$ binder-free. 
%	
%	Suppose that a term has binder $\Gamma\mid\Delta\vdash\overline{xy}.M$. We reason diagrammatically about $\Gamma,x\mid\Delta,y\vdash M$. Using the equational theory for TMCs we can redraw the diagram so that $y$ is level with the input interface and $x$ level with the output interface:
%	
%	\centerpic{1}{deftr1}
%	
%	This means that the term has form 
%	\[\Gamma,x\mid\Delta,y\vdash M\equiv (y\otimes m);M'(x\otimes n)
%	\]
%	 for some natural numbers $m,n$ and term $M'$, provable in the equational theory of the TMC. But using Prop.~\ref{prop:linktr} we can now replace the binder with a trace, equationally. Therefore, all binders can be replaced with traces \textit{equationally}. The resulting binder-free terms, since they denote equal diagrams, must be also equationally equal in the SMT theory.
%	%This is why we keep trace at 0, because trace at $m$ can be defined from trace at $0$ and $1$ using the vanishing axiom. The absence of trace at 0 would need a significantly changed equational theory, as seen in the proof of Prop.~\ref{prop:linktr}.
%\end{proof}

\section{Monoid and comonoid structures}

In the specification of diagrams corresponding to circuits and systems two kinds of structural components are used quite extensively: the splitting of two connectors and the joining of two connectors. A \textit{monoid}, in this particular diagrammatic context is a pair of constant morphisms $(\phi, \eta)$ where $\phi:2\rightarrow 1$ is the \textit{multiplication} and $\eta:1\rightarrow 0$ the \textit{unit}, interpreted as the following biflow diagrams corresponding to ``joining two connectors'' and, respectively, to a ``dangling output'' connector:
\begin{center}
	\raisebox{1.5ex}{$\sbr{\phi:2\rightarrow 1}=$}\includegraphics[]{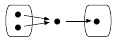} \qquad
	\raisebox{1.5ex}{$\sbr{\eta:0 \rightarrow 1}=$}\includegraphics[]{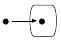}. 
\end{center}
Conversely, a \textit{co-monoid} consist of a \textit{co-multiplication} $\psi:1\rightarrow 2$ corresponding to ``splitting'' two connectors and a \textit{co-unit} $\epsilon:0\rightarrow 1$, a ``dangling input''. 
\begin{center}
	\raisebox{1.5ex}{$\sbr{\psi:1\rightarrow 2}=$}\includegraphics[]{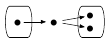} \qquad
	\raisebox{1.5ex}{$\sbr{\epsilon:1\rightarrow 0}=$}\includegraphics[]{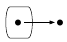}. 
\end{center}
The monoid should be associative, commutative and have a unit law. The corresponding diagrams are in Fig.~\ref{fig:mono}. The co-monoid should be co-associative, co-commutative and have a co-unit, with the equations and diagrams the converse of the above.

\begin{figure}[t]
\includegraphics[scale=0.85]{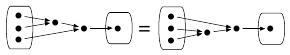}\hfill
\includegraphics[scale=0.85]{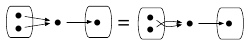}\hfill
\includegraphics[scale=0.85]{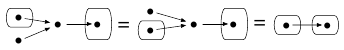}\\
$(\phi\otimes 1);\phi=(1\otimes \phi);\phi$ \hfill
$\phi=\gamma;\phi$ \hfill
$(1\otimes \eta);\phi=(\eta\otimes 1);\phi=1$
\caption{Commutative monoid laws}\label{fig:mono}
\end{figure}

Other laws that apply in some contexts are \textit{Frobenius} ($(\phi\otimes 1);(1\otimes \psi)=\phi;\psi$) and the \textit{special} law ($\psi;\phi=1 $). Diagrammatically they are:
\centerpic{1}{specfrob}
If all these equations are satisfied, any structure constructed out of the multiplication, co-multiplication, unit and co-unit can be reduced to a canonical form called informally a \textit{spider diagram}~\cite{Coecke2008}. Informally, all such constructions will denote a diagram of shape equivalent to 
\centerpic{1}{spider}
Precisely how these equations are chosen and how they are incorporated into the underlying graph model presents many choices that must be dealt with on a case-by-case basis, depending on the intended applications. In the next sequel we  consider two interesting scenarios, but more variations are possible. 

\subsection{Comonoid structure}

One natural way to add a comonoid structure to a biflow diagram is by adding constant morphisms $\psi:1\rightarrow 2$ and $\epsilon:1\rightarrow 0$ with the semantics of the previous section. The wire-homeomorphism, which consists of the removal of all possible unlabelled nodes (or, conversely, the insertion of spurious unlabelled nodes) carry over in the obvious way to the new setting. Note that in the absence of the monoid structure all wire nodes, by construction, have exactly one incoming and one outgoing edge. With the monoid structure in place, every wire node has exactly one incoming edger and zero, one or more outgoing edges and wire homeomorphisms need to re-assign the target of the edge. The FPG $(V\uplus \{a\},I,O,E,f)$ such that $f(a)$ is undefined ($a$ is a wire node) is (still) homeomorphic to the FPG
$
(V,I,O,E\setminus\{a\}\cup\{(b,c) \mid (b,a),(a,c)\in E\}).
$
Diagrammatically, the homeomorphism is represented as:

\begin{center}
\includegraphics{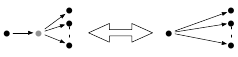}
\end{center}

The associativity, co-commutativity and co-unit laws now come then directly from the new wire homeomorphism. For example $\psi;(\psi\otimes 1)=\psi;(1\otimes\psi)$ is:
\begin{center}
	\includegraphics[]{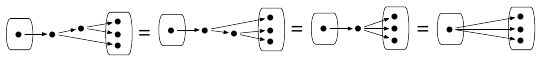}
\end{center}
%From the point of view of syntax, the presence of new structural connectors creates new opportunities for expressiveness but also for complications. In the case of circuit specifications, for example, a very common modification of a circuit that a designer might want to make is to identify an arbitrary point in the circuit as an observable output, for debugging purposes. For example, in the example circuit we used for motivation, the designer might want to observe the second output of $f$:
%\centerpic{1}{comonoidex}
%We recall that the circuit on the left was simply $f;g$. One way to write the circuit on the right is \[f;(1\otimes\psi\otimes 1);(1\otimes 1\otimes \gamma);(g\otimes 1).\] 
%The flat, nominal syntax on the other hand requires the addition just a new line to correspond to the new wire. As before, the flat syntax is more robust to change. However, we can use the link syntax to deal with comonoid structures just by tweaking the type system. 

To reflect the fact that wire nodes have now exactly one incoming and zero or more outgoing edges we update the type system to allow contraction and weakening for input variables while output variables preserve the linear discipline, as seen if Fig.~\ref{fig:comdi}.
\begin{figure}
\begin{center}
	\small 
	\AxiomC{}
	\UnaryInfC{$\Gamma, x \mid - \vdash x : 0\rightarrow 1$} 
	\DisplayProof\quad
	\AxiomC{}
	\UnaryInfC{$\Gamma\mid x  \vdash x : 1\rightarrow 0$} 
	\DisplayProof
	\quad
	\AxiomC{}
	\UnaryInfC{$\Gamma \mid - \vdash 1 : 1\rightarrow 1$} 
	\DisplayProof\\[1.5ex]
	\AxiomC{}
	\UnaryInfC{$\Gamma\mid - \vdash \gamma : 2\rightarrow 2$} 
	\DisplayProof
	\quad
	\AxiomC{}
	\UnaryInfC{$\Gamma \mid - \vdash \psi : 1\rightarrow 2$} 
	\DisplayProof\quad
	\AxiomC{}
	\UnaryInfC{$\Gamma\mid - \vdash \epsilon : 1\rightarrow 0$} 
	\DisplayProof
	\\[1.5ex]
	\AxiomC{$\Gamma \mid \Delta \vdash M: m_1\rightarrow m_2$}	
	\AxiomC{$\Gamma \mid \Delta' \vdash N: n_1\rightarrow n_2$}
	\RightLabel{$m_2=n_1$}
	\BinaryInfC{$\Gamma \mid \Delta\uplus\Delta' \vdash M;N : m_1\rightarrow n_2$}
	\DisplayProof
	\quad	
	\AxiomC{$\Gamma \mid \Delta \vdash M: m_1\rightarrow m_2$}	
	\AxiomC{$\Gamma \mid \Delta' \vdash N: n_1\rightarrow n_2$}
	\BinaryInfC{$\Gamma \mid \Delta\uplus\Delta' \vdash M\otimes N : m_1\rightarrow n_2$}
	\DisplayProof
	\\[1.5ex]
	\AxiomC{$\Gamma\mid\Delta \vdash M: m_1\rightarrow m_2$}	
	\AxiomC{$i=0,1$}
	\BinaryInfC{$\Gamma\mid\Delta \vdash \mathrm{Tr}^i (M):m_1-i\rightarrow m_2-i$}
	\DisplayProof\quad
	\AxiomC{$\Gamma,x \mid\Delta, y\vdash M: m_1\rightarrow m_2$}	
	\UnaryInfC{$\Gamma\mid\Delta  \vdash \overline{xy}.M:m_1\rightarrow m_2$}
	\DisplayProof
\end{center}
\caption{Type system for diagrams with a co-monoid.}
\label{fig:comdi}
\end{figure}
Even though the semantics of link remains the same, the graph invariants are different because of the presence of the co-multiplication and the co-unit. The diagram for $\overline{xy}.M$ is still the connection of (the single) occurrence of the node labelled with $x$ to the (zero or more) occurrences of the node(s) labelled with $y$, followed by the removal of the labels. The mathematical formulation is the same as before (Eqn.~\ref{eqn:link}), except for the anchoring relation.

\begin{lemma}[Biflow comonoid definability]\label{lem:bcdef}
	Any biflow diagram with a comonoid structure is definable in terms of composition, tensor and link. 
\end{lemma}
%\begin{proof}
%	The proof is again the same as for Lem.~\ref{lem:defdag}, with the key graph invariant being that each wire node has one incoming and several outgoing edges. We first reduce all wire nodes so that there exists exactly one wire node between any to constant morphisms. Then we rearrange the diagram so that all constant morphisms are globally set in parallel. Finally, we replace all wire nodes with two fresh variables $x,y$ and link them. 	
%	\centerpic{1}{defcom}
%	The linear use of $x$ and the possibly non-linear use of $y$ ensures that the top-level \textit{link} construct is typed correctly.
%\end{proof}
From this it follows immediately that the \textit{link} construct is enough in terms of expressiveness, making the explicit use of the co-unit and the co-multiplication optional. 
\begin{theorem}
	For any term $ M$ in the language of biflow diagrams with comonoids, there are forms $ \tilde M$ and $\hat M$ such that $\tilde M$ is link-free and $\hat M$ is free of structural combinators (identity, symmetry, trace, comultiplication, counit) and 
	$
	\sbr M=\sbr{\tilde M}=\sbr{\hat M}.
	$
\end{theorem}
The equational theory of the diagrammatic language is extended by the following two new equations:
\begin{center} 
\framebox{$
	\epsilon = \overline{xy}.x
	\qquad
	\psi = \overline{xy}.x\otimes y\otimes y
	.
	$}
\end{center}
\begin{theorem}\label{eqn:como}
	The equational theory of biflow diagrams with a comonoid is sound and complete for the given semantics. 
\end{theorem}
%\begin{proof}[Proof sketch]%
%	The proof follows the same pattern as before. We broaden the scope of the \textit{link} binder then we reason about the diagram in which variables are left uninterpreted. We manipulate the diagram to bring it to the same form as in Lem.~\ref{lem:bcdef} so that links can be replaced by a combination of comultiplication and trace. The resulting binder-free diagrams must be equal equationally in the SMT theory.
%\end{proof}
An analogous  link syntax and equational theory for a monoid structure is obvious.

\subsection{Frobenius structure}
A particularly useful combination of the monoid and the comonoid structures uses the Frobenius and special axioms, leading to so-called \textit{spider diagrams}~\cite{Coecke2008}.
Syntactically, besides the new constants $\phi$ and $\eta$, the monoid multiplication and unit, the type judgements now lose all linearity, with link nodes useable zero, one or more times both as input and output (Fig.~\ref{fig:spityp}). 

\begin{figure*}
	\begin{center}\small 
	\AxiomC{}
	\UnaryInfC{$\Gamma, x \mid \Delta \vdash x : 0\rightarrow 1$} 
	\DisplayProof\qquad
	\AxiomC{}
	\UnaryInfC{$\Gamma\mid x,\Delta  \vdash x : 1\rightarrow 0$} 
	\DisplayProof
	\qquad
	\AxiomC{}
	\UnaryInfC{$\Gamma \mid \Delta \vdash 1 : 1\rightarrow 1$} 
	\DisplayProof
	\qquad
	\AxiomC{}
	\UnaryInfC{$\Gamma \mid \Delta \vdash \phi : 2\rightarrow 1$} 
	\DisplayProof
	\\[1.5ex]
	\AxiomC{}
	\UnaryInfC{$\Gamma\mid \Delta \vdash \gamma : 2\rightarrow 2$} 
	\DisplayProof
	\quad
	\AxiomC{}
	\UnaryInfC{$\Gamma \mid \Delta \vdash \psi : 1\rightarrow 2$} 
	\DisplayProof\quad
	\AxiomC{}
	\UnaryInfC{$\Gamma\mid \Delta \vdash \epsilon : 1\rightarrow 0$} 
	\DisplayProof\quad
	\AxiomC{}
	\UnaryInfC{$\Gamma\mid \Delta \vdash \eta : 0\rightarrow 1$} 
	\DisplayProof
	\\[1.5ex]
	\AxiomC{$\Gamma \mid \Delta \vdash M: m_1\rightarrow m_2$}	
	\AxiomC{$\Gamma \mid \Delta \vdash N: n_1\rightarrow n_2$}
	\RightLabel{$m_2=n_1$}
	\BinaryInfC{$\Gamma \mid \Delta \vdash M;N : m_1\rightarrow n_2$}
	\DisplayProof\qquad
	\AxiomC{$\Gamma,x \mid\Delta, y\vdash M: m_1\rightarrow m_2$}	
	\UnaryInfC{$\Gamma\mid\Delta  \vdash \overline{xy}.M:m_1\rightarrow m_2$}
	\DisplayProof
	\\[1.5ex]	
	\AxiomC{$\Gamma \mid \Delta \vdash M: m_1\rightarrow m_2$}	
	\AxiomC{$\Gamma \mid \Delta \vdash N: n_1\rightarrow n_2$}
	\BinaryInfC{$\Gamma \mid \Delta \vdash M\otimes N : m_1\rightarrow n_2$}
	\DisplayProof
	\quad
	\AxiomC{$\Gamma\mid\Delta \vdash M: m_1\rightarrow m_2$}
	\AxiomC{$i=0,1$}	
	\BinaryInfC{$\Gamma\mid\Delta \vdash \mathrm{Tr}^i (M):m_1-i\rightarrow m_2-i$}
	\DisplayProof
\end{center}
\caption{Type system for biflow Frobenius diagrams}
\label{fig:spityp}
\end{figure*}
\begin{theorem}
	For any term $ M$ in the language of biflow diagrams with a Frobenius structure, there are forms $ \tilde M$ and $\hat M$ such that $\tilde M$ is link-free and $\hat M$ is free of structural combinators (identity, symmetry, trace, (co)multiplication, (co)unit) and 
	$
	\sbr M=\sbr{\tilde M}=\sbr{\hat M}.
	$
\end{theorem}
The two equations for the monoid structure are:
\framebox{$
	\eta=\overline{xy}.y,
	\quad
	\phi = \overline{xy}.x\otimes x\otimes y$.
}
\begin{theorem}\label{thm:eqnspd}
	The equational theory of biflow diagrams with a Frobenius structure is sound and complete for the given semantics. 
\end{theorem}
%The proof is again similar, using the fact that the spider diagram is a canonical form which can be reconstructed using the monoid and the comonoid (co) multiplications.

Note that the interaction of the monoid and co-monoid is not necessarily subject to the Frobenius law in diagrammatic languages. This is not a structure that is well studied in its own right, but it turns out to occur naturally in the treatment of digital circuits~\cite{GhicaJL}. The syntax in this case is less natural and the equational properties are more elusive, but the expressiveness and convenience of the notation remains an important attribute. 

\section{Example}

In this section we are giving an example of circuit-like system which has both structural features and arbitrary connections, namely a neural network. We will employ a free use of all the methods of specification available (TMC with monoid and co-monoid and the link operation) to give a succinct and informative description which is, in our subjective opinion, preferable to both the purely structural or the purely relational descriptions. These will be left as reader exercises. 

The net we use as an example is the \textit{gated recurrent unit}~\cite{DBLP:conf/emnlp/ChoMGBBSB14}, a simplification of the popular \textit{long short-term memory} net~\cite{DBLP:journals/neco/HochreiterS97} which is broadly used in wide range of applications, from machine translation to image classification. The basic cell of the net has the following architecture, where $\sigma$, $\tau$ are activation functions and  $\times, +$ arithmetic operations. 

\begin{center}
	\includegraphics[]{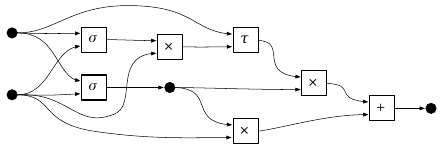}
\end{center}

Let $\psi_2:2\rightarrow 4$ be the (definable) co-monoid on a pair of inputs. 
The mixed-syntax and the purely structural notations for the GRU cell are:
\begin{align*}
\mathit{gru}&=
\overline{x'x}.\overline{ h'h}.
(x'\otimes h');
(x\otimes h);\psi_2;(\sigma \otimes h\otimes\sigma);
(x\otimes {\times} \otimes \psi);
(\tau\otimes 1\otimes {\times});
({\times}\otimes 1);
{+}\\
&=(\psi\otimes 1);(1\otimes\psi\otimes\psi);(3\otimes\psi
\otimes\psi);(2\otimes \gamma\otimes 3);(1\otimes\sigma\otimes\sigma\otimes 2);(2\otimes \gamma\otimes 1);\\
&\qquad (1\otimes{\times}\otimes\psi\otimes 1); (\tau\otimes 1\otimes{\times});({\times}\otimes 1);{+}
\end{align*}
Note that in the above the $x,h$ variables can have any width, not necessarily unit. 

A recurrent neural net consists of GRU cell with a recurrent link.  For practical application, instead of a recurrent link, a finite unfolding of the net is often used:
\begin{center}
	\includegraphics[]{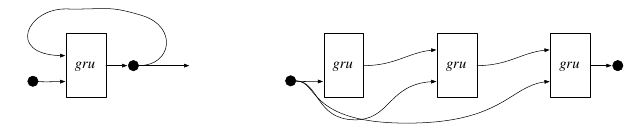}
\end{center}
The two nets can be described as $\mathrm{Tr}(\mathit{gru};\psi)$ and $\overline{x'x}. x'; (\mathit{gru}\otimes x); (\mathit{gru}\otimes x); \mathit{gru}$, respectively.

\section{Related and further work}

%Computer programs have always been seen primarily as syntactic artefacts. Code is stored as long sequences of ASCII characters, most often down to the faithful representation of whitespace -- \textit{space vs. tab} are amusing but sadly real debates. However, influential designers such as Bret Victor have long argued that code should not necessarily be viewed as text, but diagrammatically\footnote{\url{http://worrydream.com/\#!/TheFutureOfProgramming}}, as this representation is much more appealing to human cognition. The switch from textual to diagrammatic notation could be as much of an improvement as the switch from numeric (binary or hexadecimal) to textual notation. And, indeed, diagrammatic representation of \textit{code} exist in some domain-specific areas such as hardware schematics (\textsc{VHDL, Verilog}) or system representation languages (\textsc{Simulink}). But even diagrammatic languages ultimately need a syntactic representation. Current languages for diagrams are low-level, requiring a large number of names, and almost devoid of compositional structure. Our paper aims to improve this, reconciling a style of syntax emerging from the categorical study of diagrams with the more established style. 

We have shown how variables can be used within a structural framework for diagram syntax. We showed how from the point of view of expressiveness the structural and the nominal syntax can be equivalent and complementary, if the use of variables is constrained by a suitable type system. We have highlighted a neat correspondence between non-linearity and the presence of a monoid or co-monoid structure in the diagram. By extending the structural notation with the new \textit{link} construct we showed that the existing equational properties are preserved and, moreover, this new structure itself has good equational properties. 
We believe that the combination of structural and nominal syntax can sometimes improve the readability and usability of diagram languages. Although this paper is only concerned with syntax, it is part of a larger research effort aimed at creating hardware languages more similar to programming languages in terms of their categorical~\cite{GhicaJung16} and operational semantics~\cite{GhicaJL}. As further work we intend to create \textit{structural} conservative extensions of diagram languages such as VHDL, \textsc{Simulink} or \textsc{Dot}. 

Several theoretical aspects we have not studied, but remain subject for further work. The most intriguing perhaps is the categorical semantics for the \textit{link} structure itself. Similar diagrammatic structures have been conjectured in the study of higher-order $\pi$-calculus, with the input and output links modelled as adjoint profunctors and a link-like operation modelled as their composition~\cite{vicary}. A smaller issue, but practically important, is understanding the syntax and the axiomatisation of links for diagrams with monoid and comonoid structures in the absence of the Frobenius law. These structures occur in diagrams modelling systems with directed flow of information, such as digital circuits, which are well motivated by applications. Finally, the various definability results (e.g. Lem.~\ref{lem:defdag}) suggest the possibility of arriving at normal forms for various classes of diagrams. 

In terms of the graph semantics, the frame point graphs is not the only possible starting point. Cospans of graphs~\cite{rosebrugh2005generic} or hypergraphs~\cite{DBLP:conf/lics/BonchiGKSZ16} have also been used as semantics for  diagrammatic monoidal categories. These alternative presentations may arguably offer, to the theorist, more conceptual clarity. What we find attractive about FPGs is the fact that they can be presented in a concrete way, eliding much of the categorical apparatus. We believe this could make our work accessible to a broader audience. Moreover, the \textit{link} construct is particularly easy to interpret in FPGs. However, these alternative semantic frameworks along with categorical semantics are being investigated.  

\noindent
\paragraph{Acknowledgement:} This work was supported in part by EPSRC grant 	EP/P004490/1.

\bibliographystyle{eptcs}
\bibliography{diag}   

\begin{thebibliography}{10}
\providecommand{\bibitemdeclare}[2]{}
\providecommand{\surnamestart}{}
\providecommand{\surnameend}{}
\providecommand{\urlprefix}{Available at }
\providecommand{\url}[1]{\texttt{#1}}
\providecommand{\href}[2]{\texttt{#2}}
\providecommand{\urlalt}[2]{\href{#1}{#2}}
\providecommand{\doi}[1]{doi:\urlalt{http://dx.doi.org/#1}{#1}}
\providecommand{\bibinfo}[2]{#2}

\bibitemdeclare{inproceedings}{DBLP:conf/lics/AbramskyC04}
\bibitem{DBLP:conf/lics/AbramskyC04}
\bibinfo{author}{Samson \surnamestart Abramsky\surnameend} \&
  \bibinfo{author}{Bob \surnamestart Coecke\surnameend} (\bibinfo{year}{2004}):
  \emph{\bibinfo{title}{A Categorical Semantics of Quantum Protocols}}.
\newblock In: {\sl \bibinfo{booktitle}{19th {IEEE} Symposium on Logic in
  Computer Science {(LICS} 2004), 14-17 July 2004, Turku, Finland,
  Proceedings}}, pp. \bibinfo{pages}{415--425},
  \doi{10.1109/LICS.2004.1319636}.

\bibitemdeclare{inproceedings}{DBLP:conf/icfp/BjesseCSS98}
\bibitem{DBLP:conf/icfp/BjesseCSS98}
\bibinfo{author}{Per \surnamestart Bjesse\surnameend}, \bibinfo{author}{Koen
  \surnamestart Claessen\surnameend}, \bibinfo{author}{Mary \surnamestart
  Sheeran\surnameend} \& \bibinfo{author}{Satnam \surnamestart
  Singh\surnameend} (\bibinfo{year}{1998}): \emph{\bibinfo{title}{Lava:
  Hardware Design in {H}askell}}.
\newblock In: {\sl \bibinfo{booktitle}{Proceedings of the third {ACM} {SIGPLAN}
  International Conference on Functional Programming {(ICFP} '98), Baltimore,
  Maryland, USA, September 27-29, 1998.}}, pp. \bibinfo{pages}{174--184},
  \doi{10.1145/289423.289440}.

\bibitemdeclare{inproceedings}{DBLP:conf/lics/BonchiGKSZ16}
\bibitem{DBLP:conf/lics/BonchiGKSZ16}
\bibinfo{author}{Filippo \surnamestart Bonchi\surnameend},
  \bibinfo{author}{Fabio \surnamestart Gadducci\surnameend},
  \bibinfo{author}{Aleks \surnamestart Kissinger\surnameend},
  \bibinfo{author}{Pawel \surnamestart Sobocinski\surnameend} \&
  \bibinfo{author}{Fabio \surnamestart Zanasi\surnameend}
  (\bibinfo{year}{2016}): \emph{\bibinfo{title}{Rewriting modulo symmetric
  monoidal structure}}.
\newblock In: {\sl \bibinfo{booktitle}{Proceedings of the 31st Annual
  {ACM/IEEE} Symposium on Logic in Computer Science, {LICS} '16, New York, NY,
  USA, July 5-8, 2016}}, pp. \bibinfo{pages}{710--719},
  \doi{10.1145/2933575.2935316}.

\bibitemdeclare{inproceedings}{DBLP:conf/popl/BonchiSZ15}
\bibitem{DBLP:conf/popl/BonchiSZ15}
\bibinfo{author}{Filippo \surnamestart Bonchi\surnameend},
  \bibinfo{author}{Pawel \surnamestart Sobocinski\surnameend} \&
  \bibinfo{author}{Fabio \surnamestart Zanasi\surnameend}
  (\bibinfo{year}{2015}): \emph{\bibinfo{title}{Full Abstraction for Signal
  Flow Graphs}}.
\newblock In: {\sl \bibinfo{booktitle}{Proceedings of the 42nd Annual {ACM}
  {SIGPLAN-SIGACT} Symposium on Principles of Programming Languages, {POPL}
  2015, Mumbai, India, January 15-17, 2015}}, pp. \bibinfo{pages}{515--526},
  \doi{10.1145/2676726.2676993}.

\bibitemdeclare{inproceedings}{DBLP:conf/emnlp/ChoMGBBSB14}
\bibitem{DBLP:conf/emnlp/ChoMGBBSB14}
\bibinfo{author}{Kyunghyun \surnamestart Cho\surnameend}, \bibinfo{author}{Bart
  \surnamestart van Merrienboer\surnameend}, \bibinfo{author}{{\c{C}}aglar
  \surnamestart G{\"{u}}l{\c{c}}ehre\surnameend}, \bibinfo{author}{Dzmitry
  \surnamestart Bahdanau\surnameend}, \bibinfo{author}{Fethi \surnamestart
  Bougares\surnameend}, \bibinfo{author}{Holger \surnamestart
  Schwenk\surnameend} \& \bibinfo{author}{Yoshua \surnamestart
  Bengio\surnameend} (\bibinfo{year}{2014}): \emph{\bibinfo{title}{Learning
  Phrase Representations using {RNN} Encoder-Decoder for Statistical Machine
  Translation}}.
\newblock In: {\sl \bibinfo{booktitle}{Proceedings of the 2014 Conference on
  Empirical Methods in Natural Language Processing, {EMNLP} 2014, October
  25-29, 2014, Doha, Qatar, {A} meeting of SIGDAT, a Special Interest Group of
  the {ACL}}}, pp. \bibinfo{pages}{1724--1734}.
\newblock \urlprefix\url{http://aclweb.org/anthology/D/D14/D14-1179.pdf}.

\bibitemdeclare{inproceedings}{Coecke2008}
\bibitem{Coecke2008}
\bibinfo{author}{Bob \surnamestart Coecke\surnameend} \& \bibinfo{author}{Ross
  \surnamestart Duncan\surnameend} (\bibinfo{year}{2008}):
  \emph{\bibinfo{title}{Interacting Quantum Observables}}.
\newblock In: {\sl \bibinfo{booktitle}{Automata, Languages and Programming,
  35th International Colloquium, {ICALP} 2008, Reykjavik, Iceland, July 7-11,
  2008, Proceedings, Part {II} - Track {B:} Logic, Semantics, and Theory of
  Programming {\&} Track {C:} Security and Cryptography Foundations}}, pp.
  \bibinfo{pages}{298--310}, \doi{10.1007/978-3-540-70583-3\_25}.

\bibitemdeclare{article}{Cazanescu:1990:TNA:97367.97373}
\bibitem{Cazanescu:1990:TNA:97367.97373}
\bibinfo{author}{V.~E. \surnamestart C\u{a}z\u{a}nescu\surnameend} \&
  \bibinfo{author}{Gheorghe \surnamestart Stef\u{a}nescu\surnameend}
  (\bibinfo{year}{1990}): \emph{\bibinfo{title}{Towards a New Algebraic
  Foundation of Flowchart Scheme Theory}}.
\newblock {\sl \bibinfo{journal}{Fundam. Inf.}}
  \bibinfo{volume}{13}(\bibinfo{number}{2}), pp. \bibinfo{pages}{171--210}.
\newblock \urlprefix\url{http://dl.acm.org/citation.cfm?id=97367.97373}.

\bibitemdeclare{incollection}{ghica2013diagrammatic}
\bibitem{ghica2013diagrammatic}
\bibinfo{author}{Dan~R \surnamestart Ghica\surnameend} (\bibinfo{year}{2013}):
  \emph{\bibinfo{title}{Diagrammatic reasoning for delay-insensitive
  asynchronous circuits}}.
\newblock In: {\sl \bibinfo{booktitle}{Computation, Logic, Games, and Quantum
  Foundations}}, \bibinfo{publisher}{Springer}, pp. \bibinfo{pages}{52--68},
  \doi{10.1007/978-3-642-38164-5\_5}.

\bibitemdeclare{inproceedings}{GhicaJung16}
\bibitem{GhicaJung16}
\bibinfo{author}{Dan~R. \surnamestart Ghica\surnameend} \&
  \bibinfo{author}{Achim \surnamestart Jung\surnameend} (\bibinfo{year}{2016}):
  \emph{\bibinfo{title}{Categorical semantics of digital circuits}}.
\newblock In \bibinfo{editor}{Ruzica \surnamestart Piskac\surnameend} \&
  \bibinfo{editor}{Muralidhar \surnamestart Talupur\surnameend}, editors: {\sl
  \bibinfo{booktitle}{Formal Methods in Computer-Aided Design (FMCAD), 2016,
  Mountain View, California, USA}}, pp. \bibinfo{pages}{41--49},
  \doi{10.1109/FMCAD.2016.7886659}.

\bibitemdeclare{inproceedings}{GhicaJL}
\bibitem{GhicaJL}
\bibinfo{author}{Dan~R. \surnamestart Ghica\surnameend}, \bibinfo{author}{Achim
  \surnamestart Jung\surnameend} \& \bibinfo{author}{Aliaume \surnamestart
  Lopez\surnameend} (\bibinfo{year}{2017}): \emph{\bibinfo{title}{{Diagrammatic
  Semantics for Digital Circuits}}}.
\newblock In \bibinfo{editor}{Valentin \surnamestart Goranko\surnameend} \&
  \bibinfo{editor}{Mads \surnamestart Dam\surnameend}, editors: {\sl
  \bibinfo{booktitle}{26th EACSL Annual Conference on Computer Science Logic
  (CSL 2017)}}, {\sl \bibinfo{series}{Leibniz International Proceedings in
  Informatics (LIPIcs)}}~\bibinfo{volume}{82}, \bibinfo{publisher}{Schloss
  Dagstuhl--Leibniz-Zentrum fuer Informatik}, \bibinfo{address}{Dagstuhl,
  Germany}, pp. \bibinfo{pages}{24:1--24:16}, \doi{10.4230/LIPIcs.CSL.2017.24}.

\bibitemdeclare{article}{Hackney2015}
\bibitem{Hackney2015}
\bibinfo{author}{Philip \surnamestart Hackney\surnameend} \&
  \bibinfo{author}{Marcy \surnamestart Robertson\surnameend}
  (\bibinfo{year}{2015}): \emph{\bibinfo{title}{On the Category of Props}}.
\newblock {\sl \bibinfo{journal}{Applied Categorical Structures}}
  \bibinfo{volume}{23}(\bibinfo{number}{4}), pp. \bibinfo{pages}{543--573},
  \doi{10.1007/s10485-014-9369-4}.

\bibitemdeclare{inproceedings}{DBLP:conf/tlca/Hasegawa97}
\bibitem{DBLP:conf/tlca/Hasegawa97}
\bibinfo{author}{Masahito \surnamestart Hasegawa\surnameend}
  (\bibinfo{year}{1997}): \emph{\bibinfo{title}{Recursion from Cyclic Sharing:
  Traced Monoidal Categories and Models of Cyclic Lambda Calculi}}.
\newblock In: {\sl \bibinfo{booktitle}{Typed Lambda Calculi and Applications,
  Third International Conference on Typed Lambda Calculi and Applications,
  {TLCA} '97, Nancy, France, April 2-4, 1997, Proceedings}}, pp.
  \bibinfo{pages}{196--213}, \doi{10.1007/3-540-62688-3\_37}.

\bibitemdeclare{article}{DBLP:journals/neco/HochreiterS97}
\bibitem{DBLP:journals/neco/HochreiterS97}
\bibinfo{author}{Sepp \surnamestart Hochreiter\surnameend} \&
  \bibinfo{author}{J{\"{u}}rgen \surnamestart Schmidhuber\surnameend}
  (\bibinfo{year}{1997}): \emph{\bibinfo{title}{Long Short-Term Memory}}.
\newblock {\sl \bibinfo{journal}{Neural Computation}}
  \bibinfo{volume}{9}(\bibinfo{number}{8}), pp. \bibinfo{pages}{1735--1780},
  \doi{10.1162/neco.1997.9.8.1735}.

\bibitemdeclare{article}{joyal1991geometry}
\bibitem{joyal1991geometry}
\bibinfo{author}{Andr{\'e} \surnamestart Joyal\surnameend} \&
  \bibinfo{author}{Ross \surnamestart Street\surnameend}
  (\bibinfo{year}{1991}): \emph{\bibinfo{title}{The geometry of tensor
  calculus, {I}}}.
\newblock {\sl \bibinfo{journal}{Advances in Mathematics}}
  \bibinfo{volume}{88}(\bibinfo{number}{1}), pp. \bibinfo{pages}{55--112},
  \doi{10.1016/0001-8708(91)90003-P}.

\bibitemdeclare{article}{Kissinger}
\bibitem{Kissinger}
\bibinfo{author}{Aleks \surnamestart Kissinger\surnameend}
  (\bibinfo{year}{2012}): \emph{\bibinfo{title}{Pictures of processes:
  automated graph rewriting for monoidal categories and applications to quantum
  computing}}.
\newblock {\sl \bibinfo{journal}{arXiv preprint arXiv:1203.0202}}.
\newblock \urlprefix\url{https://arxiv.org/abs/1203.0202}.

\bibitemdeclare{article}{rosebrugh2005generic}
\bibitem{rosebrugh2005generic}
\bibinfo{author}{Robert \surnamestart Rosebrugh\surnameend},
  \bibinfo{author}{Nicoletta \surnamestart Sabadini\surnameend} \&
  \bibinfo{author}{RFC \surnamestart Walters\surnameend}
  (\bibinfo{year}{2005}): \emph{\bibinfo{title}{Generic commutative separable
  algebras and cospans of graphs}}.
\newblock {\sl \bibinfo{journal}{Theory and applications of categories}}
  \bibinfo{volume}{15}(\bibinfo{number}{6}), pp. \bibinfo{pages}{164--177}.
\newblock \urlprefix\url{http://www.tac.mta.ca/tac/volumes/15/6/15-06abs.html}.

\bibitemdeclare{article}{DBLP:journals/logcom/SadrzadehCC13}
\bibitem{DBLP:journals/logcom/SadrzadehCC13}
\bibinfo{author}{Mehrnoosh \surnamestart Sadrzadeh\surnameend},
  \bibinfo{author}{Stephen \surnamestart Clark\surnameend} \&
  \bibinfo{author}{Bob \surnamestart Coecke\surnameend} (\bibinfo{year}{2013}):
  \emph{\bibinfo{title}{The {F}robenius anatomy of word meanings {I:} subject
  and object relative pronouns}}.
\newblock {\sl \bibinfo{journal}{J. Log. Comput.}}
  \bibinfo{volume}{23}(\bibinfo{number}{6}), pp. \bibinfo{pages}{1293--1317},
  \doi{10.1093/logcom/ext044}.

\bibitemdeclare{incollection}{selinger2010survey}
\bibitem{selinger2010survey}
\bibinfo{author}{Peter \surnamestart Selinger\surnameend}
  (\bibinfo{year}{2010}): \emph{\bibinfo{title}{A survey of graphical languages
  for monoidal categories}}.
\newblock In: {\sl \bibinfo{booktitle}{New structures for physics}},
  \bibinfo{publisher}{Springer}, pp. \bibinfo{pages}{289--355},
  \doi{10.1007/978-3-642-12821-9\_4}.

\bibitemdeclare{misc}{vicary}
\bibitem{vicary}
\bibinfo{author}{J~\surnamestart Vicary\surnameend}: \bibinfo{note}{Personal
  communication}.

\end{thebibliography}
\end{document}